\documentclass[runningheads]{llncs}

\RequirePackage[colorlinks=true]{hyperref}
\hypersetup{
  linkcolor=[rgb]{0,0,0.4},
  citecolor=[rgb]{0, 0.4, 0},
  urlcolor=[rgb]{0.6, 0, 0}
}
\usepackage{pgfplots}
\usepackage{amsmath}
\usepackage{amsfonts}
\usepackage{lipsum}
\usepackage{setspace}
\usepackage{mdframed}
\usepackage{multicol}
\usepackage{multirow}
\usepackage{subcaption}
\usepackage{booktabs}
\usepackage{graphicx}
\usepackage{comment}
\usepackage{tikz}
\usepackage{xspace}
\usepgflibrary{arrows}
\usepackage{fancyvrb}
\usepackage{multicol}
\usepackage{comment}
\usepackage{thm-restate}

\usetikzlibrary{decorations.pathreplacing,backgrounds}
\usetikzlibrary{plotmarks}

\usepackage[font=footnotesize]{caption}

\usepackage{algorithmicx}
\usepackage{algorithm} % http://ctan.org/pkg/algorithms
\usepackage[noend]{algpseudocode}

\algrenewcommand\algorithmicindent{1.0em}%

\newcommand\sse{\subseteq}
\newcommand\Sym[1]{\ensuremath{\mathrm{Sym}_{#1}}}
\newcommand\set[1]{\ensuremath{\{#1\}}}
\newcommand\condset[2]{\set{#1 \;|\; #2}}
\newcommand\NP{\ensuremath{\mathsf{NP}}}
\newcommand\coNP{\ensuremath{\mathsf{coNP}}}

\newcommand\poly[1]{\ensuremath{\mathrm{poly}(#1)}}

\newcommand{\puz}[2]{$(#1, #2)$-puzzle\xspace}
\newcommand{\puzs}[2]{$(#1, #2)$-puzzles\xspace}
\newcommand{\susp}[2]{$(#1, #2)$-SUSP\xspace}
\newcommand{\susps}[2]{$(#1, #2)$-SUSPs\xspace}

\newcommand{\censor}[1]{#1}

\date{}

\title{Efficiently-Verifiable Strong Uniquely Solvable
  Puzzles and Matrix Multiplication}
\titlerunning{Efficiently-Verifiable SUSPs and Matrix Multiplication}
\author{
  Matthew Anderson%\inst{1}
  \and%
 Vu Le%\inst{1}
}
\institute{Department of Computer Science \\ Union College
  \\ Schenectady, New York, USA \\ \email{\{andersm2, lev\}@union.edu}}

\begin{document}

\maketitle
\begin{abstract}
Following the approach of \cite{ajx20}, we advance the Cohn-Umans
framework \cite{cu03,cksu05} for developing fast matrix multiplication
algorithms.  We introduce, analyze, and search for a new subclass of
strong uniquely solvable puzzles (SUSP), which we call
\emph{simplifiable SUSPs}.  We show that these puzzles are efficiently
verifiable, which remains an open question for general SUSPs.  We also
show that individual simplifiable SUSPs can achieve the same strength
of bounds on the matrix multiplication exponent $\omega$ that infinite
families of SUSPs can.  We report on the construction, by computer
search, of larger SUSPs than previously known for small width.
This, combined with our tighter analysis, strengthens the upper bound on
the matrix multiplication exponent from $2.66$ to $2.505$ obtainable
via this computational approach, and nears the results of the
handcrafted constructions of \cite{cksu05}.
  
\end{abstract}

\section{Introduction}
\label{sec:intro}

Square matrix multiplication is a fundamental mathematical operation:
Given $n \in \mathbb{N}$, a field $\mathbb{F}$, and matrices $A, B \in
\mathbb{F}^{n \times n}$, compute the resulting matrix $C = AB$ where
the entry $(i,k) \in [n]^2$ is $C_{i,k} = \sum_{j \in [n]} A_{i,j}B_{j,k}$.

The complexity of this problem has been well studied.  Early work by
Strassen gave a recursive, divide-and-conquer algorithm for square
matrix multiplication that runs in time $O(n^{2.81})$ \cite{str69}.
The situation steadily improved over the next two decades, culminating
with the $O(n^{2.376})$ time Coppersmith-Winograd algorithm
\cite{cw90}.  More recently, a series of refinements to the
Coppersmith-Winograd algorithm has resulted in a state-of-the-art
algorithm that runs in time $O(n^{2.37188})$
\cite{ds13,le14,aw20,dwz22}.  The question remains open: \emph{ What
is the smallest $\omega$ for which there exists a matrix
multiplication algorithm that runs in time $O(n^\omega)$?}

Instead of following the traditional approach of refinements to
Coppersmith-Winograd, we pursue the framework developed by Cohn and
Umans \cite{cu03,cksu05}.  This framework connects the existence of
efficient algorithms for matrix multiplication to the existence of
combinatorial objects called \emph{strong uniquely solvable puzzles (SUSP)}.  An \emph{\puz{s}{k}} $P$ is a subset of $\set{1,2,3}^k$ with
cardinality $|P| = s$.  We defer the formal definition of SUSPs to
\autoref{sec:prelim}, but note that on input $P$, the strong unique
solvability of $P$ is decidable in \coNP\footnote{It remains open
  whether SUSP verification is \coNP-complete.}.  The larger the
\emph{size} $s$ of a strong uniquely solvable puzzle is for a fixed
$k$, the more efficient of a matrix multiplication algorithm is implied
by the Cohn-Umans framework (see \autoref{lem:puzzle-omega}).
Anderson et al.~initiated a systematic computer-aided search for large puzzles that are SUSPs \cite{ajx20}.  They developed algorithms that are
sufficiently efficient in practice---using reductions to \NP-hard
problems, and sophisticated satisfiability and integer programming
solvers---for verifying SUSPs.  They applied those algorithms to find
large SUSPs of small width $k \le 12$.

There are several aspects of the work of Anderson et al.~that
warranted further study: (i) although the verification algorithm was
shown to be experimentally effective, its worst-case performance was
exponential time, (ii) the results they used from \cite{cksu05}
to imply efficient matrix multiplication algorithms were limited
because they only found individual SUSPs of small size, rather than
infinite families of SUSPs like in the constructions of \cite{cksu05},
and (iii) they experimentally observed that for some pairs of SUSPs
$P_1$, $P_2$, the Cartesian product $P_1 \times P_2$ was also an SUSP,
but they did not provide a theoretical explanation as to why.  These
aspects limited the small-width SUSPs that were found in
\cite{ajx20,ajx23} to only be able to achieve the bound $\omega \le
2.66$.

\subsection{Our Contributions}

We make progress on the computer-aided search for large SUSPs and
resolve the three limitations mentioned above by introducing a new
class of SUSPs that we call \emph{simplifiable SUSPs}.

In \cite{ajx20} they show that the problem of verifying whether a
puzzle $P$ is an SUSP reduces to determining whether a related
tripartite hypergraph $H_P$ has no nontrivial 3D matchings.  We
describe a polynomial-time simplification algorithm that takes a 3D
hypergraph and attempts to simplify it to the trivial matching without
changing the set of matchings the graph has.  In this way, we define
simplifiable SUSPs to be puzzles $P$ whose 3D hypergraph $H_P$
simplifies to the trivial matching.  This gives a polynomial-time
algorithm to generate a proof that $P$ is an SUSP.  In this way,
simplifiable SUSPs are polynomial-time verifiable by definition, making
them more feasible to search for.
\begin{restatable}{theorem}{simplifyefficient}
%\begin{theorem}
  \label{thm:simplify}
  Let $P$ be an \puz{s}{k}. There is an algorithm for determining whether $P$ is a simplifiable SUSP.  The algorithm runs in time $\poly{s, k}$.
%\end{theorem}
\end{restatable}

We show that simplifiable SUSPs have a number of other interesting
properties that make them a good candidate to search for when trying to
improve bounds on $\omega$.  In particular, we show that simplifiable
SUSPs are a natural generalization of \emph{local SUSPs} from
\cite{cksu05}.  Local SUSPs are also efficiently verifiable, but since
they are not densely encoded, they are hard to effectively search for.
Relatedly, we show that simplifiable SUSPs are closed under Cartesian
product, which is not the case for general SUSPs.  We show that this
property allows a single simplifiable SUSP to generate an infinite
family of SUSPs by taking all powers of the puzzle.  This allows the
stronger infinite-family bound on $\omega$ of \cite{cksu05} to be
applied, which strengthens the bounds on $\omega$ implied by
individual simplifiable SUSPs.
\begin{restatable}{theorem}{simplifiableomega}
%\begin{theorem}
  \label{thm:simplifiableomega}
  Let $\epsilon > 0$, if there is a simplifiable \susp{s}{k} $P$, then
  there is an algorithm for multiplying $n$-by-$n$ matrices in time
  $O(n^{\omega+\epsilon})$ where
  $$\omega \le \min_{m \in \mathbb{N}_{\ge 3}} 3 \cdot \frac{k \log m
    - \log s}{k \log(m-1)}.$$
  %\end{theorem}
\end{restatable}
\noindent Additionally, we show that
simplifiable SUSPs can achieve any bound on $\omega$ that SUSPs can.

Finally, we report finding new large simplifiable SUSPs of small width
that improve the bounds on $\omega$ from $2.66$ to $2.505$ via the
computational Cohn-Umans approach.  The SUSPs we construct for small
width are considerably larger than those of the previous work
\cite{cksu05,ajx20,ajx23}, and imply stronger bounds on $\omega$ for
the same domain.  However, it is important to note that this
computational approach has yet to surpass the $\omega \le 2.48$ bound
implied by the infinite families of SUSPs handcrafted in
\cite{cksu05}, or the state-of-the-art Coppersmith-Winograd
refinements with the record bound of $\omega \le 2.37188$ \cite{dwz22}.

Our results further the computational approach to developing efficient
matrix multiplication algorithms using the Cohn-Umans framework started
by \cite{ajx20}.  Although it has yet to do so, this programme is
motivated by the hope that with further advancement this
non-traditional approach might meet or even exceed the algorithms that
result from refinements to Coppersmith-Winograd.

\subsection{Related Work}

For more background on and history of algorithms for the
matrix multiplication problem, see the excellent survey by Bl{\"a}ser
\cite{gs005}.

Some negative results are known for the Cohn-Umans framework that
apply to our work as well. In particular, a series of articles
\cite{clp17,asu13,bccgu16,avw18a} showed that there exists an
$\epsilon > 0$ such that this framework, as well as a variety of other algorithmic
approaches, cannot achieve $\omega = 2 + \epsilon$.  This implies that
our approach cannot achieve the best potential result of
$\widetilde{O}(n^2)$, however, the authors are unaware of a
concrete value known for this $\epsilon$.  There remains considerable
distance between the state-of-the-art refinements of the
Coppersmith-Winograd algorithms and the known lower bounds.

Our search for simplifiable SUSPs is implemented using a standard search
technique called \emph{iterative local search}, c.f, e.g,
\cite{ai:ma4}.  Some comparison with our work can be drawn to another
recent, widely publicized computational approach by Fawzi et al.~who
used reinforcement learning to generate low-rank representations of
the matrix multiplication tensor \cite{fawzi2022}, producing
algorithms with $\omega \le 2.77$.  Although their results avoid the involved
representation-theoretic machinery of the Cohn-Umans framework, 
our $\omega$ bounds are considerably stronger than theirs, which is
also true for the earlier work in \cite{ajx20,ajx23}.

\subsection{Organization}

\autoref{sec:prelim} discusses relevant background on strong uniquely
solvable puzzles and their relationship with matrix multiplication
algorithms from \cite{cksu05}, and the connection between the
verification of SUSPs and 3D perfect matching from \cite{ajx20}.
\autoref{sec:efficient} develops some observations about 2D and 3D
matching that lead to the definition of simplifiable SUSPs, shows that
simplifiable SUSPs are efficiently verifiable, and shows that they are
a generalization of local SUSPs.  \autoref{sec:infinite-families}
proves several useful properties of simplifiable SUSPs, including that
they generate infinite families of SUSPs and as a consequence imply
stronger bounds on $\omega$.  \autoref{sec:results} reports on the
new large SUSPs we found, the concrete bounds on $\omega$ they imply compared to previous work, and briefly discusses our search
algorithms and implementation.  \autoref{sec:conclusion} concludes
with several related open problems.

\section{Preliminaries}
\label{sec:prelim}

For a natural number $n \in \mathbb{N}$, we use $[n]$ to denote the set
$\{1, 2,...,n\}$. $\Sym{Q}$ denotes the symmetric group on the
elements of a set $Q$.

\begin{definition}[Puzzle]
  For $s, k \in \mathbb{N}$, an \puz{s}{k} is a subset $P \subseteq [3]^k$
  with $|P| = s$.
\end{definition}
We say that an \puz{s}{k} has $s$ rows and $k$ columns. The columns
are inherently ordered and indexed by $[k]$. The rows are not
inherently ordered, although it is often convenient to assume that
they are arbitrarily ordered and indexed by $[s]$.  Cohn et
al.~studied the following class of puzzles that we call \emph{SUSPs}
\cite{cksu05}.

\begin{definition}[Strong Uniquely Solvable Puzzle (SUSP)]
  \label{def:strong}
  An $(s, k)$-puzzle $P$ is \emph{strong uniquely solvable} if $\forall \pi_1, \pi_2, \pi_3 \in Sym_P$, either (i) $\pi_1 = \pi_2 =
  \pi_3$, or (ii) $\exists r \in P$ and $i \in [k]$ such that exactly two of the following conditions are true: $(\pi_1(r))_i = 1$,
  $(\pi_2(r))_i = 2$, $(\pi_3(r))_i = 3$.
\end{definition}
Based on \autoref{def:strong}, the task of determining whether a
puzzle is an SUSP is in \coNP.  Anderson et al.~studied the problem of
determining whether a puzzle is an SUSP, devised a reduction from
this problem to a variant of the 3D perfect matching problem, and then
used it to develop a practical, but worst-case exponential time,
algorithm \cite{ajx20}.

Cohn et al.~also considered the following subset of SUSPs, called
\emph{local SUSPs}, which are puzzles that naturally demonstrate that
they are SUSPs.
\begin{definition}[Local SUSP]
  \label{def:local}
  An \puz{s}{k} $P$ is \emph{local strong uniquely solvable} if for
  each $(u, v, w) \in P^3$ with $u, v, w$ not all equal, there exists $c
  \in [k]$ such that $(u_c, v_c, w_c)$ is in the set
  $$\mathcal{L} = \set{(1,2,1), (1,2,2), (1,1,3), (1, 3, 3), (2, 2,
    3), (3, 2, 3)}.$$
\end{definition}
Based on \autoref{def:local}, the task of determining whether a puzzle
is a local SUSP can be done in time $O(s^3 \cdot k)$, by checking all
triples of rows.  Cohn et al.~show that SUSPs can be converted to
local SUSPs, albeit with a substantial increase in the parameters.
\begin{proposition}[{\cite[Proposition 6.3]{cksu05}}]
  \label{prop:susp-to-local}
  Let $P$ be an \susp{s}{k}, then there is a local \susp{s!}{sk}
  $P'$.  Moreover, SUSP capacity is achieved by local SUSPs.
\end{proposition}
Note that the second consequence of this proposition is that any bound
on $\omega$ that can be achieved by SUSPs can be achieved by local
SUSPs.

\subsection{From Matrix Multiplication to SUSPs} 

Using the concept of an SUSP, \cite{cu03} showed how to define group
algebras that allow matrix multiplication to be efficiently embedded
into them.  The existence of SUSPs implies upper bounds on the
matrix multiplication exponent $\omega$. 

The \emph{SUSP capacity} is defined as the largest constant $C$
such that there exist SUSPs of size $(C - o(1))^k$ and width $k$ for
infinitely many values of k \cite{cksu05}.  The constructions of Cohn
et al.~produce families $\mathcal{F}$ of \susps{s(k)}{k} for
infinitely many values of $k$.  The key parameter that relates $\omega$ to
the size of puzzles is the \emph{capacity $C_\mathcal{F}$ of the family}, defined as the limit of $(s(k))^{\frac{1}{k}}$ as $k$
goes to $\infty$.  Cohn et al.~showed the following bound on $\omega$
as a function of capacity.

\begin{lemma}[{\cite[Corollary 3.6]{cksu05}}]
  \label{lem:family-omega}
  Let $\epsilon > 0$, if there is a family $\mathcal{F}$ of SUSPs with
  capacity $C_\mathcal{F}$, then there is an algorithm for multiplying
  $n$-by-$n$ matrices in time $O(n^{\omega+\epsilon})$ where
  $$\omega \le \min_{m \in \mathbb{N}_{\ge 3}} 3 \cdot \frac{\log m -
    \log C_\mathcal{F}}{\log(m-1)}.$$
\end{lemma}
In the same corollary, Cohn et al.~showed a weaker bound on $\omega$
based on a single SUSP.
\begin{lemma}[{\cite[Corollary 3.6]{cksu05}}]
  \label{lem:puzzle-omega}
  Let $\epsilon > 0$, if there is an \susp{s}{k}, there is an
  algorithm for multiplying $n$-by-$n$ matrices in time
  $O(n^{\omega+\epsilon})$ where
  $$\omega \le \min_{m \in \mathbb{N}_{\ge 3}} 3 \cdot  \frac{sk  \log
    m - \log s!}{sk\log(m-1)}.$$
\end{lemma}
Cohn et al.~also shows that if the SUSP capacity is
$C_{\mathrm{max}} = 3 / 2 ^ {2/3}$, it immediately follows that
$\omega = 2$.  As mentioned in the Introduction, subsequent work has
shown that the SUSP capacity is strictly less than $C_{\mathrm{max}}$.
That said, SUSPs still represent a viable route to improving the
efficiency of matrix multiplication algorithms.

\subsection{From SUSPs to 3D Matchings}

Let $G$ be a $r$-uniform hypergraph over $r$ disjoint copies of a domain
$U$.  We only consider $r \in \set{2, 3}$ and use ``2D graph'' to
refer to the case where $r = 2$ and ``3D graph'' to refer to the case
where $r = 3$.  We use the notation $V(G)$ to denote the vertex set of
$G$ and $E(G)$ to denote the edge set of $G$.  We say that $G$ has a
\emph{perfect matching} if there exists $M \sse E(G)$ such that $|M| =
|U|$ and for all distinct pairs of edges $a, b \in M$, $a$ and $b$ are
vertex disjoint, that is, $a_i \neq b_i, \forall i \in [r]$.  Note
that we only consider perfect matchings in this article, so often drop
``perfect'' for brevity.  The \emph{trivial matching} of $G$ is the
set $\condset{u^r}{u \in U}$.  We call a matching $M$
\emph{ nontrivial} if it is not the trivial matching of $H_P$.

For two $r$-partite graphs $G_1, G_2$ over domains $U_1$ and $U_2$,
respectively, we define their \emph{tensor product} to be the
$r$-partite graph $G_1 \times G_2$ over the Cartesian product of their
domain sets $U_1 \times U_2$, and whose edges are the Cartesian
product of their edge sets $$E(G_1) \times E(G_2) =
\condset{((u_1,u_2),(v_1, v_2))}{(u_1, v_1) \in E(G_1), (u_2,v_2) \in
  E(G_2)}.$$ We note that the adjacency matrix of the tensor product
of two $r$-partite graphs is the Kronecker product of the two adjacency
matrices of the graphs; this perspective is helpful in
visualizing some of our results from \autoref{sec:efficient}.

Anderson et al.~showed a reduction from checking whether an \puz{s}{k}
$P$ is an SUSP to deciding whether there are no nontrivial perfect
matchings in a related 3D graph $H_P$ \cite{ajx20}.  We briefly
recall that construction.  Define a function $f$ to represent the
inner condition of \autoref{def:strong} on triplets of rows $u, v, w
\in P$ where $f(u, v, w) = 1$, if $\exists i \in [k]$ such that exactly two
of the following hold: $u_i = 1, v_i = 2, w_i = 3$ and $f(u, v, w) =
0$, otherwise. Then, they define $H_P$ to be the 3D graph with domain
$P$ whose edges are $E(H_P) = \condset{(u, v, w)}{f(u, v, w) = 0}$.
Note that the trivial matching is a matching of $H_P$.

With these definitions in hand, we state the main result of
\cite{ajx23} that we need.
\begin{lemma}[{\cite[Lemma 5]{ajx23}}]
  \label{lem:susp-to-3dm}
A puzzle $P$ is an SUSP iff $H_P$ has no nontrivial perfect
matchings.
\end{lemma}
In general, deciding whether a 3D graph has a perfect matching is
$\NP$-complete \cite{karp72}.

\section{Simplification and Efficiently-Verifiable SUSPs}
\label{sec:efficient}

The reduction from SUSP verification to the problem of 3D perfect
matching, from \autoref{lem:susp-to-3dm}, leads to a na\"{i}ve
worse-case $O(2^s \cdot \poly{s,k})$-time algorithm for
verification. This approach was not effective in practice, so in
\cite{ajx20}, they solved this 3D perfect matching instance by further
transforming it into a mixed-integer programming problem and then
applying a powerful commercial solver. Here we introduce a useful
subset of SUSPs that are efficiently verifiable to overcome this
limitation.

Let $P$ be an \puz{s}{k} and $H_P$ be its corresponding 3D graph
as in \autoref{lem:susp-to-3dm}. If $H_P$ has a non-trivial matching,
the matching itself witnesses this fact.  However, if $H_P$ has no
non-trivial matchings, there does not need to be a short witness of
this fact (the widely held conjecture that $\NP \neq \coNP$ supports
this view).  The subclass of SUSPs we develop naturally has short
witnesses.

Our approach is based on the following insight about the 3D graph
$H_P$: If $H_P$ has a matching, the matching projects to three 2D
matchings of the 2D faces of $H_P$.  Moreover, if edges in one of the faces
cannot be used for a matching of that face, none of the edges of $H_P$
that project onto that edge can be used in a 3D matching of $H_P$.  We
iteratively apply this idea to efficiently \emph{simplify} the 3D
graph $H_P$, without changing the matchings it has, until it is
reduced to a trivial matching or no further simplification can be
made.  If the 3D graph is reduced to the trivial matching, it means
that $H_P$ had no nontrivial matchings, and the puzzle $P$ must be an
SUSP.  We call such puzzles \emph{simplifiable SUSPs}.  A by-product of this
simplification process is a series of edges deletions of $H_P$, which
provides a witness that $P$ is an SUSP.

\subsection{Simplifying 2D Graphs}

We build up to the simplification of 3D graphs and the definition of
simplifiable SUSPs by first looking at the analogous situation for 2D
graphs.  The following lemma gives a way to remove certain edges from
2D graphs without eliminating matchings.

\begin{figure}[t]
  \begin{center}
    \includegraphics[width=.3\linewidth]{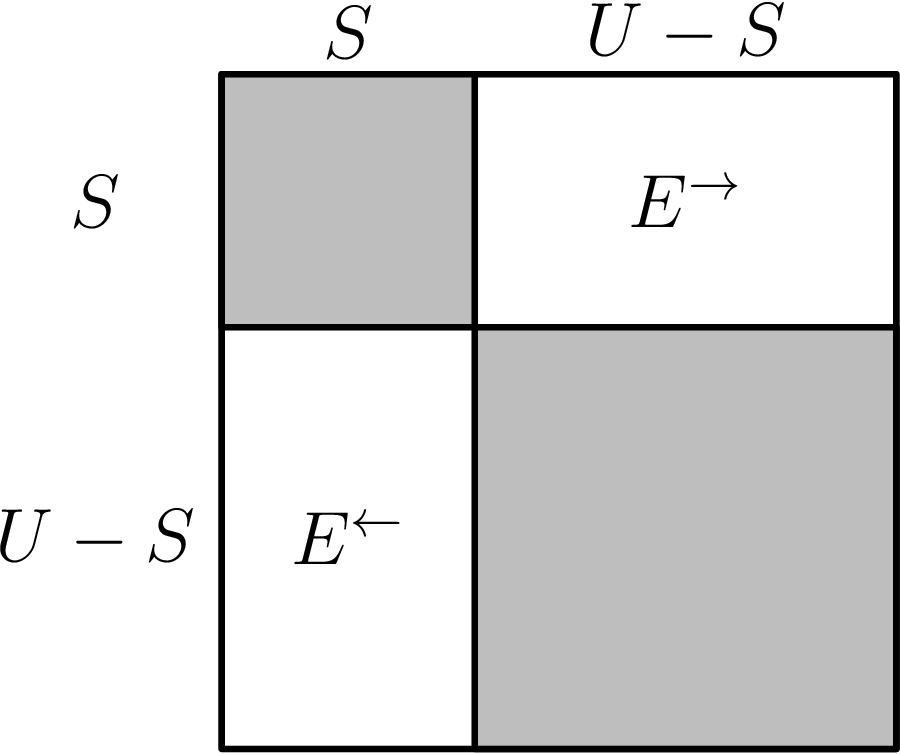}
  \end{center}
  \caption{Let $G$ be 2D graph over the domain $U$.  This diagram
    represents the partitioning of the adjacency matrix of $G$
    relative to a set $S \sse U$ which divides the adjacency matrix
    into four regions of edges, $S \times S$, $S \times (U-S)$, $(U-S)
    \times S$, $(U-S) \times (U-S)$.  The edges in the gray regions
    survive the simplification to $G'$ as in
    \autoref{lem:simplify_set}, while any edges in $E^\rightarrow$ or
    $E^\leftarrow$ are deleted from $G$.\label{fig:simplify_set}}
\end{figure}

\begin{lemma}
  \label{lem:simplify_set}
  Let $G$ be a 2D graph with domain $U$.  Let $S \sse U$, $E^\rightarrow
  = S \times (U - S)$, and $E^\leftarrow =(U - S) \times S$.  Let $G'$
  be a 2D graph with domain $U$ and edges $E(G') = E(G) -
  E^\rightarrow - E^\leftarrow$.  If $E^\rightarrow \cap E(G) =
  \emptyset$ or $E^\leftarrow \cap E(G) = \emptyset$, then $G'$ has the same set of perfect matchings as $G$.
\end{lemma}
To get an intuitive sense for why this lemma holds,
\autoref{fig:simplify_set} visualizes the adjacency matrix of $G$,
showing it divided it into four regions depending on $S \sse U$.  If
$G$ has no edges in one of $E^\rightarrow$ or $E^\leftarrow$, any
matching $M$ of $G$ must match $S$ to $S$ and $(U - S)$ to $(U - S)$.
Therefore, dropping $E^\rightarrow$ and $E^\leftarrow$ when constructing
$G'$ does not remove any matchings.
\begin{proof}[Proof of \autoref{lem:simplify_set}]
  Observe that since the edges of $G'$ are a subset of the edges of $G$, $G'$ cannot have a matching that $G$ does not have.  It remains to show that for each perfect matching $M$ of $G$, $M$ is also a
  perfect matching of $G'$.

  Let $M \sse E(G)$ be a perfect matching of $G$.  There are two cases
  to consider. Suppose $E^\rightarrow \cap E(G) = \emptyset$.
  Consider an edge $(u, v) \in M$. If $u \in S$, then $v \notin (U-S)$
  since there are no edges in $G$ that intersect with $S \times
  (U-S)$.  Therefore, $v \in S$.  Thus, for each $u \in S$, $(u, v) \in M$
  and $v \in S$, so $M$ matches $S$ to $S$.  If $u \in (U-S)$ and $(u,
  v) \in M$, then $v \notin S$ since for all $v \in S$ there already
exists a one-to-one correspondence with $u' \in S$ where $(u', v)
  \in M$.
  
  Thus, $M$ must match $S$ to $S$ and match $U-S$ to $U-S$, that is, $M
  \sse (S \times S) \cup ((U - S) \times (U - S))$.  Hence, $M$ must be a perfect matching of $G'$, because $M \cap (E^\rightarrow \cup
  E^\leftarrow) = \emptyset$ and therefore the edges in $M$ are deleted. The case when $E^\leftarrow \cap E(G) = \emptyset$ is symmetric. \qed
\end{proof}
Let $S \sse U$ be a subset of vertices in a 2D graph $G$ with domain $U$
for which the conditions of \autoref{lem:simplify_set} are met.  We
say that $S$ \emph{induces a simplification} of $G$ to $G'$.  We now
consider sequences of such simplifications.

\begin{definition}
  \label{def:simplifies_to}
  Let $G_0, G_1, \ldots, G_\ell$ be a sequence of 2D graphs with a common domain $U$ and let $S_1, S_2, \ldots, S_\ell \sse U$ be sets such that $S_i$ induces a simplification of $G_{i-1}$ to $G_i$ for $1 \le i \le \ell$.  We say that $G_0$ \emph{ simplifies to} $G_\ell$.
\end{definition}
The following is a corollary resulting from repeated application of
\autoref{lem:simplify_set} to the sets and 2D graphs in the above
definition.

\begin{corollary}
  \label{cor:simplifies_to}
  Let $G, G'$ be 2D graphs over the same domain.  If $G$ simplifies to $G'$, then $G$ and $G'$ have the same set of perfect matchings.
\end{corollary}

\begin{proof}
  Suppose $G$ simplifies to $G'$. By \autoref{def:simplifies_to},
  there exists $G_0, G_1, \ldots, G_\ell$ with $G = G_0$ and $G' =
  G_\ell$ and sets $S_1, S_2, \ldots, S_\ell$ for which $S_i$ induces a simplification of $G_{i-1}$ to $G_i$.  Using \autoref{lem:simplify_set}, between $G_{i-1}$ and $G_i$, one can show, by induction, that the set of perfect matchings for all $G_i$
  are the same.  Therefore, $G = G_0$ and $G' = G_\ell$ have the same set
  of perfect matchings. \qed
\end{proof}
The above sequence of arguments can be generalized to show that
simplification can be applied to graphs that are tensor products.
\begin{corollary}
  \label{cor:simplifies_to_prod}
  Let $G, G'$ be 2D graphs over the same domain $U$, and $F$ be a 2D graph over a different domain $V$.  If $G$ simplifies to $G'$, then $G
  \times F$ simplifies to $G' \times F$.
\end{corollary}
\begin{proof}[Proof Sketch]
  Let $S_1, S_2, \ldots, S_\ell \sse U$ and $G = G_0, G_1, G_2,
  \ldots, G_\ell = G'$ be the series of sets and graphs that witness $G$ simplifying to $G'$.  One can argue that the sets $S_1 \times V,
  S_2 \times V, \ldots, S_\ell \times V$ induce the corresponding chain of simplifications $G \times F = G_0 \times F, G_1 \times F,
  \ldots, G_\ell \times F = G' \times F$.  The argument for the
individual simplification steps here proceeds analogously to the proof of \autoref{lem:simplify_set}. \qed
\end{proof}
  
\subsection{Simplifying 3D Graphs}

We lift the notion of simplification from 2D graphs to 3D graphs.
Consider a 3D graph $H$ with domain $U$.  We construct three 2D graphs
$R_0, R_1, R_2$, on the same domain $U$, which, respectively, correspond
to projecting out the first, second, and third coordinates of $H$.  In
particular, the edges of these 2D graphs are, respectively,
%\begin{equation*}
$$
\begin{aligned}
  E(R_0) &= \condset{(v, w)}{\exists u \in U, (u,
    v, w) \in E(H)}, \\
  E(R_1) &= \condset{(u, w)}{\exists v \in U, (u,
    v, w) \in E(H)}, \\
  E(R_2) &= \condset{(u, v)}{\exists w \in U, (u,
    v, w) \in E(H)}.
\end{aligned}
$$
%\end{equation*}
If $H$ has a perfect matching, then it projects into a perfect matching
for each of the $R_f$'s.  To see this, let $M$ be a perfect matching
of $H$, then following the projection, define $M_0 = \condset{(v,
  w)}{\exists u \in U, (u, v, w) \in M}$.  By definition $M_0 \sse
E(R_0)$.  Because $M$ is a perfect matching of $H$, $\condset{v}{(u, v,
  w) \in M} = \condset{w}{(u, v, w) \in M} = U$, and $|M| = |U|$, so
$M_0$ is a perfect matching of $R_0$.  The argument for $R_1$ and
$R_2$ is analogous.  Furthermore, one can argue that if a matching is
nontrivial for $H$, then it is nontrivial for at least two of the
$R_f$'s.

We observe that simplifications induced on any of $R_0$, $R_1$, $R_2$,
also induce a simplification of $H$.  For brevity, the result below
is stated only for $R_0$, but holds similarly for $R_1$ and $R_2$ using
symmetric arguments.

\begin{lemma}
  \label{lem:simplify_lift}
  Let $H$, $R_0$, $U$ be defined as above.  Let $H'$ be the 3D graph
  over the domain $U$ whose edges are $E(H') = E(H) - U \times ((S
  \times (U - S)) \cup ((U - S) \times S))$.  If $S \sse U$ induces a simplification of $R_0$, then $H'$ has the same set of perfect
  matchings that $H$ does.
\end{lemma}

\begin{proof}
  Observe that since the edges of $H'$ are a subset of the edges of $H$, $H'$ cannot have a matching that $H$ does not have.  It remains to show that for each matching $M$ of $H$, $M$ is also a matching of $H'$.
  
  Let $M$ be a matching of $H$.  Suppose, for the sake of
contradiction, that $M$ is not a matching of $H'$.  There must exist an edge $(u, v, w) \in M$ that lies in the set of edges deleted in $H'$.  Let $M_0$ be the projection of $M$ into $R_0$, so that $M_0$ is a matching of $R_0$ and $(v, w) \in M_0$.  By hypothesis and definition of $H'$, $(v, w) \in (S \times (U - S)) \cup ((U - S) \times S)$.  This is a contradiction to the fact that $S$
  simplifies $R_0$, because, by \autoref{lem:simplify_set}, $(S \times
  (U - S)) \cup ((U - S) \times S)$ does not intersect with any
  matchings of $R_0$. \qed
\end{proof}

When the conditions of \autoref{lem:simplify_lift} are met, we say that
this set $S$ \emph{induces a simplification of $H$ via $R_0$}.  As
before, we can lift the notion of simplification to a series of induced
simplifications.  Here it is more complex because changing $H$ changes
its projections.  Let $S_1, S_2, \ldots, S_\ell \sse U$ and $f_1, f_2,
\ldots, f_\ell \in \set{0, 1, 2}$. We define a series of tuples of
graphs $(H_j, R_{0,j}, R_{1,j}, R_{2,j})$ with $0 \le j \le \ell$,
where $H_0 = H$, $R_{0,0} = R_0, R_{1,0} = R_1, R_{2,0} = R_2$ and for
$j > 0$, $R_{f_j,j}$ is the simplification of $R_{f_j, j-1}$ induced
by $S_j$, $H_j$ is the simplification of $H_{j-1}$ induced by $S_j$
via $R_{f_j}$ and $R_{(f_j + 1 \mod 3), j}$ and $R_{(f_j + 2 \mod 3),
  j}$ are the result of reprojecting $H_j$.  For brevity in describing
this situation, we say that $H$ \emph{simplifies to} $H_\ell$.  As
before, repeated application of \autoref{lem:simplify_lift} and
\autoref{lem:simplify_set} implies that $H_\ell$ has the same set of
matchings as $H_0 = H$ does and results in the following corollary.
\begin{corollary}
  \label{lem:simplify_pm}
  Let $H$, $H'$ be 3D graphs with the same domain.  If $H$ simplifies to $H'$, then $H$ and $H'$ have the same set of perfect matchings.
\end{corollary}
Similarly to \autoref{cor:simplifies_to_prod}, the simplification of 3D
graphs lifts to tensor products.
\begin{corollary}
  \label{cor:simplify_3dm_prod}
  Let $H$, $H'$ be 3D graphs with the same domain and $K$ be a 3D graph with a different domain.  If $H$ simplifies to $H'$, then $H
  \times K$ simplifies to $H' \times K$.
\end{corollary}

\subsection{Simplifiable SUSPs}

We now apply the notion of simplification to help in checking whether
an \puz{s}{k} $P$ is an SUSP. By \autoref{lem:simplify_pm},
$H_P$ has a nontrivial matching iff any simplification of $H_P$ has a
nontrivial matching.  This suggests a way to construct a witness that
$P$ is an SUSP: If $H_P$ simplifies to the trivial matching, then, by
\autoref{lem:simplify_pm}, $H_P$ has no nontrivial matchings, and, by
\autoref{lem:susp-to-3dm}, $P$ is an SUSP.  The sequence of sets and
their corresponding projection indexes are a witness that $P$ is an
SUSP.  Moreover, if we exclude simplifications that do not change the
3D graph, the number of edges in the 3D graph---at most $s^3$---is a
limit on the number of simplification steps that can occur.
\begin{definition}[Simplifiable SUSP]
  \label{def:simplifiable}
  An \puz{s}{k} $P$ is a \emph{simplifiable} SUSP if $H_P$ simplifies to the trivial 3D perfect matching.
\end{definition}
By definition, simplifiable SUSP are SUSPs with short ($O(s^4)$ bit length)
witnesses.  To make this definition effective, we describe a
polynomial-time algorithm that simplifies puzzles.  In particular, the
algorithm takes $H_P$; projects it onto its 2D faces, $R_0$, $R_1$,
$R_2$; then, for each face, determines sets that induce maximal
simplification of the faces; and, finally, applies those simplifications to
$H_P$ to form a new 3D graph $H'_P$.  The algorithm repeats this until
a fixed point is reached.  The resulting 3D graph is the fully
simplified version of $H_P$.  If that simplified graph is the
trivial matching, this process witnesses that $P$ is a (simplifiable)
SUSP.  For completeness, this process is described in
\autoref{alg:simplify}.

\begin{algorithm}[t]
  \caption{: \textproc{Simplify}}
  \label{alg:simplify}
  \begin{algorithmic}[1]
    \Require{A 3D graph $H$.}
    \Ensure{A fully-simplified 3D graph.}
    \Function{Simplify}{$H$}
    \State $R_0, R_1, R_2 = \Call {Project}{H}$ \label{algo:simplify:line:proj}
    \State $f \gets 0$
    \State $sinceChange \gets 0$
    \While {$sinceChange < 3$}
    \State $edgesToRemove \gets \Call{CalcEdgesToRemove}{R_f}$ \label{algo:simplify:line:2DM}
    \For{$(u, v) \in edgesToRemove$} \label{algo:simplify:line:update:start}
    \If {$f = 0$} 
    \State Delete all edges $(*, u, v)$ from $H$
    \ElsIf {$f = 1$}
    \State Delete all edges $(u, *, v)$ from $H$
    \ElsIf {$f = 2$} 
    \State Delete all edges $(u, v, *)$ from $H$
    \EndIf
    \EndFor \label{algo:simplify:line:update:end}  
    \If {$edgesToRemove = \emptyset$}
    \State $sinceChange \gets sinceChange + 1$
    \Else 
    \State $sinceChange \gets 0$
    \State $R_0, R_1, R_2 = \Call{Project}{H}$ \label{algo:simplify:line:update-project}
    \EndIf
    \State $f \gets (f + 1)\mod 3$
    \EndWhile
    \State \Return $H$
    \EndFunction
  \end{algorithmic}
\end{algorithm}

In \autoref{alg:simplify}, the subroutine \textproc{Project} takes the
3D graph $H$ and returns three 2D graphs $R_0, R_1, R_2$ that,
respectively, correspond to projecting out the first, second, and
third coordinates of $G$, as defined above.  This subroutine can be
na\"{i}vely implemented in $O(s^3)$ time.

The subroutine \textproc{CalcEdgesToRemove} at Line
\ref{algo:simplify:line:2DM} takes each of the 2D graphs corresponding
to the faces and returns a list of edges that are not used in any
maximum 2D matchings of that face.  This subroutine can be implemented
using the algorithm described in \cite[Algorithm 2]{tassa12} (also,
\cite{regin94}).  Their algorithm works by constructing the strongly
connected components of the input 2D graph $R_f$, when $R_f$ is viewed as
a directed graph over $P$ rather than a bipartite graph over $P \sqcup
P$.  The strongly connected components calculated by this algorithm
inherently partition the vertex set $P = S_1 \cup S_2 \cup \ldots \cup
S_\ell$.

Collapsing the 2D graph $R_f$ down to its strongly connected
components leaves us with a directed graph $G_f$ with $V(G_f) = \set{v_1,
  v_2, \ldots, v_\ell}$ and $E(G_f) = \condset{(v_i, v_j)}{\exists u \in
  S_i, w \in S_j \text{ such that } (u, v) \in E(R_f)}$ with $\ell$
vertices $v_j$, one for each strongly connected component $S_j$.
Furthermore, $G_f$ must be an acyclic graph, otherwise the strongly
connected components would have been larger.  These strongly connected
components are sets that induce the simplification of $R_f$.  Let $v_j$ be
a vertex in $G_f$ that has some incident edges but that has either no
incoming or no outgoing edges.  The latter property is sufficient to
apply \autoref{lem:simplify_set} and implies that $S_j$ induces a
simplification of $R_f$.  Furthermore, this simplification corresponds to
deleting all of the edges of $v_j$ in $G_f$.

This process can be repeated until there are no more edges in
$G_f$. Note that because $G_f$ is acyclic, it will always be possible to
find such a vertex $v_j$ as long as there are edges remaining.  This
series of strongly connected components induces a complete simplification
of $R_f$.  This simplification is used to remove the corresponding
edges in the 3D graph $H$ in Lines
\ref{algo:simplify:line:update:start}-\ref{algo:simplify:line:update:end}.
The 3D graph $H$ is fully simplified when no edge can be removed from
any of the three faces.  By \cite{tassa12}, the remaining edges in
each of the projections $R_f$ are ``maximally matchable'' in that they
used in some perfect matching of $R_f$.  Thus, once this happens,
there can be no additional sets that can induce simplifications in any
of the $R_f$ that remove edges in $R_f$ (or in $H$).

Since each edge of $H$, except for the diagonal, can be removed at
most once, the algorithm must reach a fixed point within $3(|P|^3 -
|P|)$ iterations of the main loop.  The cost to update $H$ and the
projections in Lines
\ref{algo:simplify:line:update:start}-\ref{algo:simplify:line:update:end}
\& \ref{algo:simplify:line:update-project} can be amortized, with
careful bookkeeping, to cost $O(|P|^3)$ across the whole algorithm.

For 2D graphs whose domain is the \puz{s}{k} $P$, the subroutine of
\cite{tassa12} runs in $O(s^{2.5}/\sqrt{\log s})$ time.  Combining
the above analysis, the overall complexity of $\textproc{Simplify}$ is
$O(s^3 + s^3 \cdot s^{2.5} / \sqrt{\log s}) = O(s^{5.5}/\sqrt{\log
  s})$.  The results of the above arguments can be summarized in the
following lemma.
\begin{lemma}
  \label{lem:simplify}
  Let $H$ be a 3D graph over $P$.  In $\poly{|P|}$ time,
  $\Call{Simplify}{H}$ computes the complete simplification of $H$.
  Therefore, $H$ has the same set of matchings as $\Call{Simplify}{H}$.
\end{lemma}

By \autoref{def:simplifiable}, the 3D graph $H_P$ associated with a
simplifiable SUSP $P$ simplifies to the trivial matching.  Furthermore, by
\autoref{lem:simplify}, $\Call{Simplify}{H_P}$ computes, in polynomial
time, the complete simplification of $H_P$, preserving the matchings.  These
two facts imply a polynomial-time algorithm to determine whether a
puzzle $P$ is a simplifiable SUSP.

\simplifyefficient*

\begin{proof}
Perform the polynomial-time reduction from SUSP verification to 3D
matching of \cite{ajx20} to produce the 3D graph $H_P$ in
time $\poly{s,k}$.  Compute $H_P' = \Call{Simplify}{H_P}$ in time
$\poly{s}$.  In time $O(s^3)$ verify and return whether or not $H_P'$
is the trivial matching $\condset{(u, u, u)}{u \in P}$.  The
algorithm is correct by \autoref{lem:susp-to-3dm} and
\autoref{lem:simplify}. \qed
\end{proof}

It is clear from the construction that simplifiable SUSPs are a subset of
SUSPs, but simplifiable SUSPs are also a generalization of the notion of
local SUSPs from \cite{cksu05}.

\begin{lemma}
  \label{lem:local-is-simplifiable}
  Every local SUSP $P$ is a simplifiable SUSP.
\end{lemma}

\begin{proof}
  By \autoref{def:local}, for every triple of rows $u,v,w \in P$,
  there is a column $c$ such $(u_c, v_c, w_c) \in \mathcal{L}$.  This implies, by the construction of $H_P$, that $(u, v, w)$ is not an edge in $H_P$.  Taken together, this implies that $H_P$ has no edges except where $u = v = w$.  Therefore, $H_P$ is the trivial matching and explicitly satisfies \autoref{def:simplifiable} without taking any simplification steps.  We conclude that $P$ is a simplifiable SUSP. \qed
\end{proof}

Intuitively, simplifiable SUSPs are an intermediate class between
local SUSPs and SUSPs.  The sets containments are proper. There exist
SUSPs that are not simplifiable and simplifiable SUSPs that are not
local.  For example, $P_1 = \set{2233, 1232, 1123, 3311}$ is an SUSP,
but it is not a simplifiable SUSP, and $P_2 = \set{11, 23}$ is a
simplifiable SUSP, but it is not a local SUSP.  Simplifiable SUSPs
have the efficient verification of local SUSPs, but the compactness of
representation of general SUSPs--these two properties make the
prospect of searching for large simplifiable SUSPs more feasible.

\section{Simplifiable SUSPs Generate Infinite Families of SUSPs}
\label{sec:infinite-families}

\newcommand{\conc}{\circ}

We show that simplifiable SUSPs have an additional useful property,
also common to local SUSPs: They induce infinite families of SUSPs
without a loss in capacity.
%% This property means that individual simplifiable SUSPs permit
%% tighter analysis that implies stronger bounds on $\omega$ than
%% individual general SUSPs.

\subsection{Puzzle \& Family Capacity}

As mentioned in \autoref{sec:prelim}, Cohn et al.~derived bounds for
the running time of matrix multiplication using infinite families of
SUSPs (\autoref{lem:family-omega}) and individual SUSPs
(\autoref{lem:puzzle-omega}).

The first bound produces stronger results than the second.  To see
this, we define the \emph{capacity} of an \susp{s}{k} $P$ to be $C_P =
s^{\frac{1}{k}}$, this is analogous to the definition of capacity for
families of SUSPs mentioned in \autoref{sec:prelim}.  Now,
consider an SUSP $P$ and an infinite family $\mathcal{F}$ with the
same capacity $C_P = C_\mathcal{F}$.  \autoref{lem:puzzle-omega} gives
a weaker upper bound on $\omega$ for the single puzzle than
\autoref{lem:family-omega} does for the infinite family.  For example,
a \susp{14}{6} has capacity $14^{\frac{1}{6}}$ and the bound on
$\omega$ from \autoref{lem:puzzle-omega} using the dimensions of the puzzle
is $\omega \le 2.73$ and although \autoref{lem:family-omega} does
not apply, if we were to use the capacity of the puzzle instead of its
dimensions, we get $\omega \le 2.52$.  The difference between 2.73
and 2.52 is substantial considering the historical progress on
$\omega$.

\subsection{Generating Infinite Families}

We show that simplifiable SUSPs can be turned into an infinite
family of simplifiable SUSPs by taking Cartesian products (powers) of
$P$ with itself.  The resulting family has the same capacity as $P$.
This allows \autoref{lem:family-omega} to be applied, instead of
\autoref{lem:puzzle-omega}, to produce a bound on $\omega$ using the
capacity of $P$.  This reduces the gap described above so that a
simplifiable \susp{14}{6} implies $\omega \le 2.52$, instead of
$\omega \le 2.73$.

We now spell out the construction in more detail.  Let $P_1$ be an
\puz{s_1}{k_1} and $P_2$ be an \puz{s_2}{k_2}.  We define the product
of $P_1$ and $P_2$ to be the Cartesian product of their underlying
sets: $P_1 \times P_2 = \condset{r_1 \conc r_2}{r_1 \in P_1, r_2 \in
  P_2}$.  Observe that $P_1 \times P_2$ is an \puz{s_1 \cdot s_2}{k_1
  + k_2}.  Furthermore, if $P$ is an \puz{s}{k}, its $m$-th power is the
Cartesian product of $P$ with itself $m$ times, $P^m$, and observe
that this is an \puz{s^m}{k \cdot m}.  For a puzzle $P$, we can define
the infinite family $\mathcal{F}_P = \condset{P^m}{m \in \mathbb{N}}$.
Observe that $\mathcal{F}_P$ has capacity $(s^m)^{\frac{1}{k \cdot m}}
= s^\frac{1}{k}$ matching the capacity of $P$.

\begin{definition}
  An SUSP $P$ \emph{generates an infinite family of SUSPs}, if every puzzle in $\mathcal{F}_P$ is an SUSP.
\end{definition}

Unfortunately, the SUSP property is not generally preserved under
Cartesian product or powering.  For example, $P = \set{2233, 1232,
  1123, 3311}$ is an SUSP, but $P \times P$ is not.  This is a
minimum-size counterexample---there is no SUSP $P'$ with fewer rows or
columns than four where $P'^2$ is not an SUSP.  Note that we determined
this by exhaustively searching for such SUSP.  A consequence of this is
that not every SUSP generates an infinite family of SUSPs.  Although
SUSPs are generally not closed under powering, we show that
simplifiable SUSPs are.  The proof is a direct consequence of
\autoref{def:simplifiable} and \autoref{cor:simplify_3dm_prod}.

\begin{lemma}
\label{lem:simplifiable-susp-closure}
Let $P_1, P_2$ be simplifiable SUSPs, then $P_1 \times P_2$ is a
simplifiable SUSP.
\end{lemma}

\begin{proof}
  We first note that the transformation of puzzles to 3D graphs is a homomorphism, i.e., $H_{P_1 \times P_2} = H_{P_1} \times H_{P_2}$.
  By \autoref{def:simplifiable} and since $P_1$ and $P_2$ are simplifiable SUSPs, $H_{P_1}$ simplifies to the trivial matching
  $M_1 = \condset{(u,u,u)}{u \in P_1}$, and $H_{P_2}$ simplifies to
  the trivial matching $M_2 = \condset{(u, u, u)}{u \in P_2}$.  In two
  applications of \autoref{cor:simplify_3dm_prod}, we can simplify
  $H_{P_1 \times P_2} = H_{P_1} \times H_{P_2}$ to $M_1 \times
  H_{P_2}$, then to $M_1 \times M_2$.  Finally, we observe that $M_1
  \times M_2$ is the trivial matching of $H_{P_1 \times P_2}$, therefore, $H_{P_1 \times P_2}$ simplifies to the trivial matching.  Therefore,
  by \autoref{def:simplifiable}, $P_1 \times P_2$ is a simplifiable
  SUSP. \qed
\end{proof}
As an easy corollary, simplifiable SUSP generate infinite families.
\begin{corollary}
  Let $P$ be a simplifiable SUSP, $P$ generates an infinite family of simplifiable SUSPs.
\end{corollary}
Combining this corollary with \autoref{lem:family-omega} we produce a
tighter bound on $\omega$ from simplifiable SUSPs, which proves our
main theorem, restated below.
\simplifiableomega*
Although it is not the case that every SUSP generates an infinite
family, there is evidence in both experimental results of \cite{ajx20,ajx23}
and some of the puzzle constructions of \cite{cksu05} that there are
(non-local) SUSP of maximum size for their width that generate
infinite families.  For example, \cite[Proposition 3.1]{cksu05} gives
an infinite family with capacity $\sqrt{2}$ that is generated by the
\susp{2}{2} $\set{12, 33}$.

Finally, we argue that the consideration of simplifiable SUSPs does
not inherently lead to weaker bounds on $\omega$ than SUSPs.

\begin{lemma}
  \label{lem:simplifiable-achieve-susp-capacity}
  The SUSP capacity is achieved by SUSPs that are simplifiable.
\end{lemma}

\begin{proof}
  By \autoref{lem:local-is-simplifiable}, every local SUSP is a simplifiable SUSP.  By \autoref{prop:susp-to-local}, the SUSP capacity is achieved by local SUSP, and hence the SUSP capacity is also achieved by simplifiable SUSPs. \qed
\end{proof}

Since we are using \autoref{prop:susp-to-local} to construct a (local)
simplifiable SUSP from an SUSP, the size of the simplifiable SUSP is
much larger than the SUSP.  We conjecture that there is a much
tighter relationship between the sizes of simplifiable SUSPs and
SUSPs.

\begin{conjecture}
  \label{conj:size}
  If there exists an \susp{s}{k}, there exists a simplifiable \susp{s}{k}.
\end{conjecture}

\section{New Lowers Bounds on Maximum SUSP Size}
\label{sec:results}

The features of simplifiable SUSPs we proved in the previous sections
make them well suited for discovery via computer search.  We use
iterative local search (ILS) techniques to locate large simplifiable
SUSPs with small width $k \le 12$.  We find simplifiable SUSPs that
match or exceed the size of those found in previous work
\cite{cksu05,ajx20}.  Because these puzzles are simplifiable,
\autoref{thm:simplifiableomega} implies that these simplifiable SUSPs
produce much stronger bounds on $\omega$ than the SUSPs of previous work
for $k \le 12$.

\subsection{New Limits on SUSP Size}

\begin{table}[t]
  \begin{centering}
    \hspace{2.5ex}
    \begin{tabular}{ccrrrrrrrrrrrr}
      \toprule
      &&\multicolumn{12}{c}{$k$} \\ \cmidrule{3-14}
      
      & & \multicolumn{1}{c}{1} & \multicolumn{1}{c}{2} &
      \multicolumn{1}{c}{3} & \multicolumn{1}{c}{4} &
      \multicolumn{1}{c}{5} & \multicolumn{1}{c}{6} &
      \multicolumn{1}{c}{7} & \multicolumn{1}{c}{8} &
      \multicolumn{1}{c}{9} &\multicolumn{1}{c}{10} &
      \multicolumn{1}{c}{11} & \multicolumn{1}{c}{12}\\ [.5ex] \hline 

      \multirow{2}{*}{\cite{cksu05}} &$s \ge$ & 1 & 2 & 3 & 4 & 4 & 10
      & 10 & 16 & 36 & 36 & 36 & 136 \\[0.5ex]
      
      &$\omega \le$ & 3.00 & 2.88 & 2.85 & 2.85 & & 2.80 & & & 2.74 &
      & & 2.70 \\[0.5ex]
      
      \multirow{2}{*}{\cite{ajx20}}&$s \ge$ & 1 & 2 &
      3 & 5 & 8 & 14 & 21 & 30 & 42 & 64 &
      112 & 196\\[0.5ex]
      
      &$\omega \le$ & 3.00 & 2.88 & 2.85 & 2.81 & 2.78 & 2.74 &
      2.73 & 2.72 & 2.72 & 2.71 & 2.68 & 2.66 \\[0.5ex] \hline

      \multirow{2}{*}{Us}&$s \ge$ & 1 & 2 & 3 & 5 & 8 & 14 &
      \textbf{23} & \textbf{35} & \textbf{52} & \textbf{78} &
      \textbf{128} & 196\\[0.5ex]
      
      &$\omega \le$ & 3.00 & \textbf{2.67} & \textbf{2.65} &
      \textbf{2.59} & \textbf{2.57} & \textbf{2.52} &
      \textbf{2.505} & \textbf{2.52} & \textbf{2.53} &
      \textbf{2.53} & \textbf{2.52} & \textbf{2.52} \\ \bottomrule
      
    \end{tabular}
  \end{centering}
  \medskip
  \caption{Comparison with \cite{cksu05,ajx20} on lower bounds for the
    maximum of size of width-$k$ SUSPs and upper bounds on
    $\omega$ they imply.  All the results in this work are simplifiable SUSPs.
    Previous work was analyzed using \autoref{lem:puzzle-omega}, and simplifiable
SUSPs were analyzed using \autoref{thm:simplifiableomega}.  The bold font indicates improvements over prior work. \label{table:compare}}
\end{table}

\autoref{table:compare} shows our constructive improvements over
\cite{cksu05,ajx20} on the maximum size of \susps{s}{k} for $1 \le k
\le 12$.  For $k \le 5$, the sizes in \cite{ajx20} were shown to be
maximum by exhaustive search, our results match this.  For $6 < k \le
11$, we construct larger SUSPs than in the previous work.  The
simplifiable \susp{196}{12} is constructed as the square of a
simplifiable \susp{14}{6} using
\autoref{lem:simplifiable-susp-closure}.  We note that this was also
how \cite{ajx20} constructed their \susp{196}{12}, but our notion of
simplification gives a theoretical explanation for why taking product
of SUSPs can produce a SUSP.  We include some of the maximal
simplifiable SUSPs we found in \autoref{app:simplifiable-susps} and
note that these puzzles can be checked for correctness by applying
\textproc{Simplify} to the 3D graphs induced by the simplifiable SUSPs
we found.  

To compute $\omega$ we use \autoref{thm:simplifiableomega}, because
all the puzzles we construct are simplifiable SUSPs.  This results
in substantial improvements over previous work: decreasing the bound on
$\omega$ by about $0.2$ in the domain we consider.  We note that the
improvement of bounds on $\omega$ appears to stall for $k \ge 8$.  We
do not believe that this reflects a real limit on the size of simplifiable
SUSPs; rather it represents a barrier for our search techniques and the
large polynomial-time cost of running \textproc{Simplify} to determine
whether a puzzle is a simplifiable SUSP.  Although our results
improve substantially on \cite{cksu05} in the domain of $k \le 12$,
their construction achieves $\omega \le 2.48$ as $k \rightarrow
\infty$.

The experimental evidence in \autoref{table:compare} is consistent
with \autoref{conj:size}.  We have not found any \susps{s}{k} for
which we have not also found simplifiable \susps{s}{k}.  That said,
all the SUSPs we know at the boundary of the search space are also
simplifiable SUSPs.  Although
\autoref{lem:simplifiable-achieve-susp-capacity} implies that a
simplifiable SUSP can achieve the same capacity that infinite SUSP
families can, it does not immediately imply this conjecture, because
$k$ increases to $sk$ in that argument.

\subsection{Search \& Implementation}

This section briefly discusses the algorithm we used to search for
large simplificable SUSPs, and our implementation of it.

\subsubsection{Search}

The search space of \puzs{s}{k} is enormous---na{\"i}vely it scales as
$3^{s \cdot k}$.  In \cite{ajx23} they noted that for $k \le 5$, it is
feasible to exhaustively examine all distinct puzzles up to
symmetries.  For $k > 5$, an exhaustive search seems infeasible, so we
employ a different strategy---a variant of \emph{iterative local
  search} (c.f., e.g., the textbook \cite{ai:ma4} for general
background on this search strategy).

For the purposes of search we define the \emph{fitness} of an
\puz{s}{k}, which is represented as a function $f : [3]^{s \times k}
\rightarrow \set{0, 1, 2, \ldots s^3 - s}$.  We define $f(P) = s^3 -
|E(\Call{Simplify}{H_P})|$.  By \autoref{def:simplifiable}, $f(P)$ is
$f_{\max} = s^3 - s$ when $P$ is a simplifiable SUSP and $f(P) <
f_{\max}$ is $P$ is not a simplifiable SUSP.

At the base level, our algorithm maintains a queue $Q$ of \puzs{s}{k}
on the search frontier ordered by increasing fitness.  The algorithm
dequeues the highest fitness puzzle $P$ from $Q$ and then considers a
variety of local modifications to $P$: (i) changing the element in
cell $(i, j)$ of $P$ to a different element of $[3]$, (ii) permuting
the element values in a column or row of $P$ according to a
permutation of $\Sym{[3]}$, or (iii) pseudorandomly replacing the
contents of a row or column of $P$.  These modified puzzles are placed into
$Q$ ordered by their fitness.

When a simplifiable \susp{s}{k} $P$ is found, the algorithm outputs
it, then empties $Q$, and then enqueues all \susp{s+1}{k} puzzles that
have $P$ as their first $s$ rows, and restarts the search with this
new frontier $Q$.  To avoid being stuck in a loop, the algorithm keeps
a hash table of all the puzzles that have been examined and checks
them before inserting a puzzle into $Q$.  The algorithm can be primed to
start from a particular \puz{s}{k} $P$ or from an empty \puz{0}{k}.
The algorithm searches over fixed $k$, but increasing $s$.  In
principle, the algorithm halts when the search space has been
exhaustively searched, but this seems infeasible for $k > 5$.

\subsubsection{Implementation}
We implement our search algorithm in C/C++, which constrains it with a
number of practical considerations.  In practice, $Q$ and the hash
table must not exceed the available memory, and at that limit we chose to
drop puzzles with lowest fitness.  A consequence of this is that $Q$
can become empty even if the search space was not exhaustively
examined, and in doing so the implementation can miss simplifiable
SUSPs that exist.  Furthermore, the actual fitness function that we use is
more complex, because the worst-case running time of
\textproc{Simplify}, $O(s^{5.5})$, is still too inefficient to be run
for every puzzle.  In practice, the running time of
\textproc{Simplify} tends to be closer to $\Theta(s^3)$, because when
run on a typical puzzle, which is far from being a simplifiable SUSP,
no edges are removed from any of the three 2D faces causing the
algorithm to immediately reach a fixed point and stop.  However, to
make the fitness function more efficient, we use the heuristics from
\cite{ajx23}, such as checking whether the puzzle is even a uniquely
solvable puzzle, before bothering to run \textproc{Simplify}.  To make
effective use of available computing resources, the search algorithm
was parallelized with OpenMP to run multithreaded, and we chose to
implement \textproc{Simplify} both in the way described in
\autoref{alg:simplify} to run on the CPU, but also in a parallelized
form to run on the GPU using the CUDA computing platform.

Since CUDA operates on the single instruction, multiple data (SIMD)
paradigm, parts of the algorithm are vectorized to be
effectively accelerated.  At a high level, the algorithm can be broken
down into more easily parallelizable parts: (i) initializing the
projected faces of the input 3D graph $H$, (ii) decomposing each of
the projections into strongly connected components, (iii) calculating
edges to remove in each projected 2D graph, and (iv) finally updating
$H$ with changes.  To perform (iii), we use the parallelized strongly
connected components decomposition algorithm of \cite{bbbc11}.
We compared the performance of our parallel algorithm of
\textproc{Simplify} with the sequential implementation for puzzles up
to size 196. On average, and without substantial optimization the
CUDA version was able to achieve a speed up by a factor of 10 on the
GPU (GeForceGTX 1060 with 1280 CUDA cores) vs. a single CPU core (Intel
i5-6400 at 2.7 GHz).  Our implementations for \textproc{Simplify} and
our search algorithms, along with a command-line tool for verifying
simplifiable SUSPs, can be found in our publicly available code
repository: \url{https://bitbucket.org/paraphase/matmult-v2}.

\section{Conclusions}
\label{sec:conclusion}

We propose and analyze simplifiable SUSPs, a new subclass of strong
uniquely solvable puzzles.  We prove that simplifiable SUSPs have nice
properties: they are efficiently verifiable and generate infinite
families of SUSP that lead to tighter bounds on $\omega$ than the
na\"{i}ve analysis provides.

We report the existence of new large (simplifiable) SUSPs with width
$7 \le k \le 11$ and strengthen the bound on $\omega$ that they imply
compared to previous work.  The SUSPs we have found through computer search
are now close to producing the same bounds ($\omega \le 2.505$) as
those families of SUSP designed by human experts ($\omega \le 2.48$).

New insights into the structure of (simplifiable) SUSPs or the search
space seem necessary to advance this research program.  One of the
main bottlenecks in the search is the running time of
\textproc{Simplify}, even if it quickly reaches a fixed point, the
algorithm still spends $\Omega(s^3)$ time to construct an instance
from an \puz{s}{k} with $s \cdot k$ entries.  By design, whether a
puzzle is a simplifiable SUSPs is decidable in polynomial time, but
it remains open whether a puzzle being an SUSP is \coNP-complete.  As
noted above, we conjecture that the existence of an \susp{s}{k}
implies the existence of a simplifiable \susp{s}{k}.

\section{Acknowledgments}
\label{sec:ack}

The second author's work was funded in part by the \censor{Union College} Summer
Research Fellows Program.  Both authors acknowledge contributions from
other student researchers to various aspects of this research program:
\censor{Zongliang Ji, Anthony Yang Xu, Jonathan Kimber, Akriti Dhasmana,
Jingyu Yao, Kyle Doney, Jordan An, Harper Lyon, Zachary Dubinsky,
Talha Mushtaq, Jing Chin, Diep Vu, Hung Duong, Siddhant Deka, Baibhav
Barwal, Aavasna Rupakheti}. We also thank our anonymous reviewers for their helpful suggestions.

\bibliographystyle{plainurl}% the mandatory bibstyle
\bibliography{references}

\appendix

\pagebreak
\section{Examples of Large Simplifiable SUSPs}
\label{app:simplifiable-susps}

Below are examples of simplifiable SUSPs that are representative of the largest SUSPs we found for each $k \le 10$.

\begin{multicols*}{3}
\noindent {\small Simplifiable \susp{1}{1}:}
\begin{Verbatim}
1

\end{Verbatim}

\noindent {\small Simplifiable \susp{2}{2}:}
\begin{Verbatim}
11
23

\end{Verbatim}

\noindent {\small Simplifiable \susp{3}{3}:}
\begin{Verbatim}
111
123
322




\end{Verbatim}

\noindent {\small Simplifiable \susp{5}{4}:}
\begin{Verbatim}
2132
2221
2322
3111
3312





\end{Verbatim}

\noindent {\small Simplifiable \susp{8}{5}:}
\begin{Verbatim}
11111
12231
12312
13222
31132
32212
32223
33122
\end{Verbatim}

\noindent {\small Simplifiable \susp{14}{6}:} 
\begin{Verbatim}
213222
213321
221211
221312
231322
233211
233312
311211
311232
311331
323112
323331
331112
332232
\end{Verbatim}

\noindent {\small Simplifiable \susp{23}{7}:}
\begin{Verbatim}
2313133
1111221
2131122
2323112
2121322
1131231
1121333
1122323
2131312
1322223
2312112
1332213
1322311
2333213
1131313
2322132
2132333
2122113
1332133
1332321
2313223
1122231
2132123
\end{Verbatim}

\noindent {\small Simplifiable \susp{35}{8}:}
\begin{Verbatim}
31322111
12223111
32112311
32233311
31222121
12322121
32133221
31312321
32113321
13323321
13123131
12122231
32233231
13113331
13212212
13113312
12223312
13123122
32113222
13112132
31222332
32233332
13212113
31312313
31222313
32113313
31322123
13313123
12122223
12223323
12322133
12112233
13212233
31312233
32133233




\end{Verbatim}

\noindent {\small Simplifiable \susp{52}{9}:}
\begin{Verbatim}
111233111
111322113
111323131
111332132
111333312
112232111
112233113
112322131
112323133
112332311
113333333
121223111
121232113
121233131
121322133
121323311
121332313
121333331
123222111
122223113
122232131
122233133
122322311
122323313
122332331
122333333
211222113
211223131
211232133
211233311
211322313
211323331
211332333
212222131
212223133
212232311
212233313
212322331
222323333
221222133
221223311
221232313
221233331
221322333
222222311
222221313
222232331
222231333
312333332
221222121
312333113
222332213



\end{Verbatim}

\noindent {\small Simplifiable \susp{78}{10}:}
\begin{Verbatim}
3111312121
1111111121
2332112121
3123311121
1323211121
2312111121
3322211121
3332221121
1322321121
1133221121
1122322121
3323222221
3132321121
2313312121
3112212121
2333212121
3132122121
1123121121
3133322121
3321322121
3133212121
3333121121
3111312232
1111111232
2332112232
3123311232
1323211232
2312111232
3322211232
3332221232
1322321232
1133221232
1122322232
3323222233
3132321232
2313312232
3112212232
2333212232
3132122232
1123121232
3133322232
3321322232
3133212232
3333121232
3111312312
1111111312
2332112312
3123311312
1323211312
2312111312
3322211312
3332221312
1322321312
1133221312
1122322312
3323222311
3132321312
2313312312
3112212312
2333212312
3132122312
1123121312
3133322312
3321322312
3133212312
3333121312
2111312331
3323121331
2332212331
3331212331
3122211331
3132322331
1122121331
3121311331
2312312331
1323322331
3132212331
1323222332
\end{Verbatim}
\end{multicols*}

\end{document}